\documentclass[11pt]{article}

\usepackage[T1]{fontenc}
\usepackage[utf8]{inputenc}
\usepackage{lmodern}
\usepackage[letterpaper,margin=1in]{geometry}
\usepackage{amsmath,amssymb,amsthm,mathtools}
\usepackage{algorithm}
\usepackage{algorithmic}
\usepackage{tikz}
\usepackage{xcolor}
\usetikzlibrary{arrows.meta,decorations.pathreplacing,positioning}
\usepackage{microtype}
\usepackage[numbers,sort&compress]{natbib}
\usepackage[hidelinks]{hyperref}

\hypersetup{
  pdftitle={Sequential Phragmen Guarantees 2-Approximate Core Stability},
  pdfauthor={Jianan Lin and Chenhao Wang}
}
\newtheorem{theorem}{Theorem}
\newtheorem{lemma}{Lemma}
\newtheorem{corollary}{Corollary}

\theoremstyle{definition}
\newtheorem{definition}{Definition}
\theoremstyle{remark}

\newcommand{\wt}{\operatorname{wt}}

\title{Sequential Phragm\'en Guarantees 2-Approximate Core Stability}
\author{Jianan Lin$^{1}$\quad Chenhao Wang$^{2,3}$\\[0.75em]
$1$ Rensselaer Polytechnic Institute\\
$2$ Beijing Normal University-Zhuhai\\
$3$ Beijing Normal-Hong Kong Baptist University
}
\date{}

\begin{document}

\maketitle

\begin{abstract}
We study approximately stable committee selection under approval preferences, following Cheng et al.~[EC~2019] and Jiang, Munagala, and Wang~[STOC~2020].
In an approval-based committee election, $n$ voters each approve a subset of the candidates, and a committee of size $k$ is selected.
A size-$k$ committee is $\lambda$-stable if, for every nonempty subset $T$ of candidates, the number of voters who strictly prefer $T$ is less than $\lambda$ times the proportional share $\frac{n|T|}{k}$.
We consider sequential Phragm\'en, a classical polynomial-time voting rule, and prove that it always returns a $2$-stable committee.
This implies that the $2$-core is always nonempty and improves upon the $3.651$-stability guarantee of Gao, Sun, and Vondr\'ak~[EC~2026].
The guarantee also extends to weighted voters and is tight even for unweighted elections: for every $1\le\lambda<2$, sequential Phragm\'en can return a committee that is not $\lambda$-stable.
Compared with the exact-core algorithm developed in concurrent work by Becker, Greger, and Peters~[arXiv~2026], sequential Phragm\'en is simpler to implement, and our proof of its $2$-stability is elementary.
\end{abstract}


\section{Introduction}

Approval-based committee voting is a standard model for selecting a fixed-size committee from voters' preferences~\cite{aziz2017}.
Each voter approves a subset of the candidates and prefers committees containing more of her approved candidates.
Applications include parliamentary elections and organizational committee selection~\cite{brill2024}.
A central concern is proportional representation: groups of voters should receive representation in proportion to their size.
Core stability expresses this requirement by allowing a group to propose an alternative committee using its proportional share of the seats.
If every member of the group strictly prefers this alternative to the elected committee, the group can object to the outcome.
A committee is in the core if no group can make such an objection.
This condition applies to all groups of voters, including those whose approval sets differ.

Aziz et al.~\cite{aziz2017} introduced the core in approval-based committee voting, and whether it is always nonempty remained open for many years.
This question led to work on both exact core existence in special cases and approximate stability.
The approximation considered here increases the group size required to support a deviation by a factor $\lambda\ge1$.
Cheng et al.~\cite{cheng2020} proved the existence of stable lotteries over committees under approval preferences.
Jiang, Munagala, and Wang~\cite{jiang2020} developed an iterative rounding framework that gives a $32$-stable committee for arbitrary monotone preferences and, combined with the stable-lottery result, a $16$-stable committee for approval preferences.
They also conjectured that the factor $32$ for general monotone preferences could be improved to $2$.
Song and Nguyen~\cite{song2026ordinal} improved the guarantee for general monotone preferences to $6.24$ using a Lindahl equilibrium with ordinal preferences.
Gao, Sun, and Vondr\'ak~\cite{gao2025} improved the guarantee for approval preferences to $3.651$ using a Lindahl equilibrium and randomized rounding.

We prove that the classical sequential Phragm\'en rule always returns a $2$-stable committee, improving the $3.651$ guarantee to $2$.
The rule runs in polynomial time and selects one candidate at a time by distributing loads among voters.
The guarantee also extends to weighted voters, with coalition size measured by total weight.
We also prove that the factor $2$ is tight for this rule: for every $1\le\lambda<2$, there is an unweighted election in which its output is not $\lambda$-stable.
In concurrent work, Becker, Greger, and Peters~\cite{becker2026core} proved that the exact core is always nonempty and gave a polynomial-time algorithm based on harmonic entropy.
Their algorithm uses local search over committees and optimizes a payment system for each committee.
Sequential Phragm\'en is simpler to implement, and our proof of its $2$-stability is elementary.

\paragraph{Our Technique.}
Sequential Phragm\'en selects candidates one at a time, sharing each candidate's unit load among its supporters so as to keep the maximum load as low as possible~\cite{janson2016,brill2024}.
Our upper bound follows the constraints imposed by candidates that remain unelected.
Integrating these constraints over the load levels bounds the sum of squared gaps between each voter's successive load levels.
For a group that prefers an alternative committee, we account separately for candidates shared with the elected committee and for newly proposed candidates.
The loads assigned to the shared candidates, together with the squared-gap bound and Cauchy--Schwarz, show that the group is too small to support a $2$-blocking deviation.

For tightness, we adapt Skowron's construction for the proportionality degree of sequential Phragm\'en~\cite{skowronproportionality}.
A group of voters jointly approves an alternative set of candidates, but each member also approves a different set of candidates supported by voters outside the group.
These outside voters make the rule repeatedly select one candidate for each group member in turn.
Every member receives exactly $r$ approved winners, although all of them approve the same $r+1$ unelected candidates.
By choosing the group sizes and the number of seats appropriately, the size of the group supporting this alternative approaches twice the proportional share required for a deviation.

\paragraph{Additional Related Work.}
The proportionality properties of sequential Phragm\'en have been extensively studied.
Brill et al.~\cite{brill2024} showed that sequential Phragm\'en satisfies proportional justified representation (PJR) and can be computed in polynomial time.
Peters and Skowron~\cite{petersskowron2020} further studied sequential Phragm\'en through alternative notions of proportionality, including laminar proportionality and priceability, and demonstrated that Phragm\'en and Proportional Approval Voting (PAV) capture different forms of proportional representation.
The proportionality guarantees of sequential Phragm\'en have also been studied quantitatively through the notion of proportionality degree~\cite{skowronproportionality}.
Another approximation to core stability requires each member of a blocking coalition to improve her utility by a given factor.
Peters and Skowron~\cite{petersskowron2020} showed that PAV satisfies a factor-$2$ utility-based approximation to the core and that the Method of Equal Shares provides an $O(\log k)$ guarantee.
These guarantees differ from the \(\lambda\)-stability considered here, which instead relaxes the coalition-size threshold.

Exact core existence was previously established for several restricted classes of approval elections.
Peters~\cite{peters2025} proved existence when the committee has at most eight seats or there are at most fifteen candidates.
Becker, Greger, and Peters~\cite{becker2026} proved existence for at most seven voter types.
Berker et al.~\cite{berker2026} obtained further existence results for restricted instances using automated reasoning.
Goel, Jiang, and Zhou~\cite{goel2026nash} studied the Nash core and its payment assignments, and used fractional solutions to establish core existence for elections with at most eight equally weighted voters.
Core stability and its approximations have also been studied in more general models of public-good and committee selection.
Fain, Munagala, and Shah~\cite{fain2018} considered indivisible public goods with additive utilities and developed approximations to the core.
Munagala et al.~\cite{munagala2022} considered budget-constrained committee selection with more general utility functions and established approximate-core guarantees for additive and monotone submodular utilities.
For fractional committees, Kroer and Peters~\cite{kroer2026lindahl} gave a convex program for computing Lindahl equilibria with funding caps.

\section{Model}\label{sec:model}

We consider the model with weighted voters~\cite[Remark~2.2]{janson2016}.
Let $\mathcal{C}$ be a finite nonempty set of candidates and let $\mathcal{N}$ be a finite nonempty set of voters.
Each voter $i\in\mathcal{N}$ has a strictly positive real weight $w_i$ and an approval set $A_i\subseteq\mathcal{C}$.
Write
\[
  n=\sum_{i\in\mathcal{N}}w_i>0,
  \qquad
  \wt(S)=\sum_{i\in S}w_i\quad(S\subseteq\mathcal{N}).
\]
We assume that every candidate is approved by at least one voter, that is, $\bigcup_{i\in\mathcal{N}}A_i=\mathcal{C}$.
Candidates approved by no voter can be removed if at least $k$ candidates remain.
Otherwise, selecting all approved candidates and filling the remaining seats arbitrarily gives a committee in the exact core, since every voter receives her maximum possible utility.
The unweighted case has $w_i=1$ for every voter $i$.

Fix an integer $k$ with $1\le k\le |\mathcal{C}|$.
A committee is a set $W\subseteq\mathcal{C}$ with $|W|=k$.
For every set $D\subseteq\mathcal{C}$, voter $i$ has approval utility
\[
  u_i(D)=|A_i\cap D|.
\]
An instance is specified by $(\mathcal{C},\mathcal{N},(w_i,A_i)_{i\in\mathcal{N}},k)$.
A group of voters may object to a committee by proposing an alternative set of candidates that each member strictly prefers.

\begin{definition}[Approximate Stability and Core]\label{def:stability}
For $\lambda\ge1$, a nonempty alternative $T\subseteq\mathcal{C}$ is a \emph{$\lambda$-blocking deviation} from a size-$k$ committee $W$ if
\[
  \wt\bigl(\{i\in\mathcal{N}:u_i(T)>u_i(W)\}\bigr)
  \ge\lambda\frac{n|T|}{k}.
\]

$W$ is \emph{$\lambda$-stable} if it admits no $\lambda$-blocking deviation, namely, for every nonempty $T\subseteq\mathcal{C}$,
\begin{equation}\label{eq:lambda-stability}
  \wt\bigl(\{i\in\mathcal{N}:u_i(T)>u_i(W)\}\bigr)
  <\lambda\frac{n|T|}{k}.
\end{equation}
The \emph{$\lambda$-core} is the set of all size-$k$ committees that are $\lambda$-stable.
The \emph{core}, also called the \emph{exact core}, is the $1$-core.
Clearly, any $\lambda$-blocking deviation $T$ satisfies $1\le |T|\le k$.
\end{definition}

\section{Weighted sequential Phragm\'en}\label{sec:rule}

We use the load, or water-filling, formulation of sequential Phragm\'en~\cite{janson2016,brill2024}.
Let $\ell_i$ be voter $i$'s load per unit of weight, so that her total load is $w_i\ell_i$.
For a load vector $\ell=(\ell_i)_{i\in\mathcal{N}}$ and a candidate $c$, define
\[
  F_{c,\ell}(x)=\sum_{i:c\in A_i}w_i\max\{0,x-\ell_i\}.
\]
The \emph{funding height} $\theta_{\ell}(c)$ of candidate $c$ is the unique solution of
\begin{equation}\label{eq:funding-height}
  F_{c,\ell}(\theta_{\ell}(c))=1.
\end{equation}

We can view each voter as a vessel of width $w_i$ with water at height $\ell_i$.
To fund a candidate $c$, raise the water levels of the voters who approve $c$ to a common height $x$, leaving any level already above $x$ unchanged.
The total amount of water added is $F_{c,\ell}(x)$, so $\theta_{\ell}(c)$ is the first height at which these voters can jointly pay the candidate's unit cost.
At each step the rule selects a candidate with the lowest funding height and updates its supporters' loads.
Algorithm~\ref{alg:seq-phragmen} gives the procedure, and Figure~\ref{fig:water-filling} illustrates the calculation of a funding height.

\begin{algorithm}[H]
\caption{Weighted sequential Phragm\'en}\label{alg:seq-phragmen}
\begin{algorithmic}[1]
\REQUIRE An election
  $(\mathcal{C},\mathcal{N},(w_i,A_i)_{i\in\mathcal{N}})$ and a committee
  size $k$.
\ENSURE A committee $W\subseteq\mathcal{C}$ with $|W|=k$.
\STATE $W\leftarrow\varnothing$ and $\ell_i\leftarrow0$ for every
  $i\in\mathcal{N}$.
\WHILE{$|W|<k$}
  \STATE Compute $\theta_{\ell}(c)$ from~\eqref{eq:funding-height} for every
    $c\in\mathcal{C}\setminus W$.
  \STATE Choose any
    $c\in\arg\min_{d\in\mathcal{C}\setminus W}\theta_{\ell}(d)$.
  \STATE $W\leftarrow W\cup\{c\}$.
  \FORALL{$i\in\mathcal{N}$ with $c\in A_i$}
    \STATE $\ell_i\leftarrow\max\{\ell_i,\theta_{\ell}(c)\}$.
  \ENDFOR
\ENDWHILE
\STATE \textbf{return} $W$.
\end{algorithmic}
\end{algorithm}

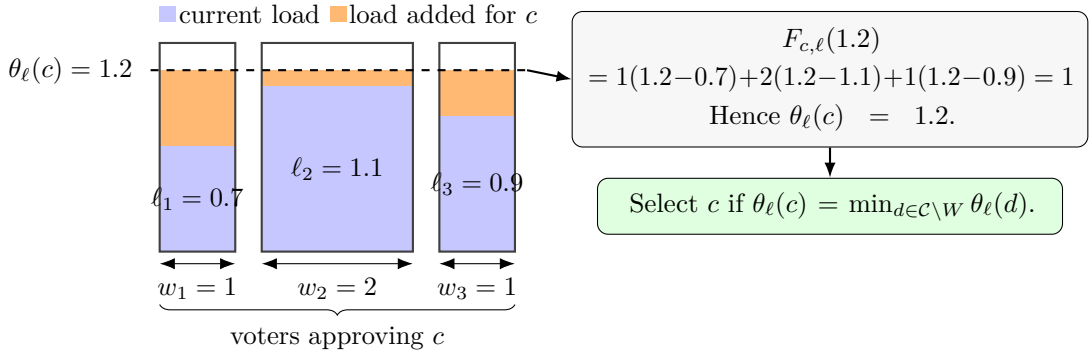
\begin{figure}[ht]
\centering
\small
\begin{tikzpicture}[
  x=1cm,
  y=2cm,
  font=\small,
  >={Latex[length=2mm]},
  vessel/.style={draw=black!75,thick},
  current/.style={fill=blue!22},
  added/.style={fill=orange!55}
]
  \fill[current] (0,0) rectangle (1,0.7);
  \fill[added]   (0,0.7) rectangle (1,1.2);
  \draw[vessel] (0,0) rectangle (1,1.38);

  \fill[current] (1.35,0) rectangle (3.35,1.1);
  \fill[added]   (1.35,1.1) rectangle (3.35,1.2);
  \draw[vessel] (1.35,0) rectangle (3.35,1.38);

  \fill[current] (3.70,0) rectangle (4.70,0.9);
  \fill[added]   (3.70,0.9) rectangle (4.70,1.2);
  \draw[vessel] (3.70,0) rectangle (4.70,1.38);

  \draw[dashed,thick] (-0.12,1.2) -- (4.86,1.2)
    node[pos=0,left=3pt] {$\theta_{\ell}(c)=1.2$};

  \node at (0.50,0.35) {$\ell_1=0.7$};
  \node at (2.35,0.55) {$\ell_2=1.1$};
  \node at (4.20,0.45) {$\ell_3=0.9$};

  \draw[<->] (0,-0.08) -- (1,-0.08)
    node[midway,below=2pt] {$w_1=1$};
  \draw[<->] (1.35,-0.08) -- (3.35,-0.08)
    node[midway,below=2pt] {$w_2=2$};
  \draw[<->] (3.70,-0.08) -- (4.70,-0.08)
    node[midway,below=2pt] {$w_3=1$};
  \draw[decorate,decoration={brace,mirror,amplitude=4pt}]
    (0,-0.35) -- (4.70,-0.35)
    node[midway,below=5pt] {voters approving $c$};

  \fill[current] (0,1.52) rectangle (0.22,1.62);
  \node[anchor=west,inner sep=0pt] at (0.25,1.57) {current load};
  \fill[added] (2.25,1.52) rectangle (2.47,1.62);
  \node[anchor=west,inner sep=0pt] at (2.50,1.57) {load added for $c$};

  \node[
    draw,
    rounded corners,
    fill=black!3,
    align=center,
    text width=6.4cm,
    inner sep=6pt,
    anchor=west
  ] (funding) at (5.45,1.14) {%
    $F_{c,\ell}(1.2)$\\[2pt]
    $=1(1.2-0.7)+2(1.2-1.1)+1(1.2-0.9)=1$\\[3pt]
    Hence $\theta_{\ell}(c)=1.2$.
  };
  \draw[->,thick] (4.86,1.2) -- (funding.west);

  \node[
    draw,
    rounded corners,
    fill=green!12,
    align=center,
    text width=5.8cm,
    inner sep=5pt,
    below=4mm of funding
  ] (choice) {Select $c$ if
    $\theta_{\ell}(c)=\min_{d\in\mathcal C\setminus W}\theta_{\ell}(d)$.};
  \draw[->,thick] (funding) -- (choice);
\end{tikzpicture}
\caption{\small One water-filling step of weighted sequential Phragm\'en.  Vessel
width is voter weight, height is voter load, and the orange area is the
additional weighted load used to fund candidate $c$.}
\label{fig:water-filling}
\end{figure}

We first record the basic properties of the load updates.

\begin{lemma}\label{lem:invariants}
Every funding height exists and is unique.
The selected funding heights are positive and nondecreasing, and every step increases the total weighted load by exactly one.
If $L$ is the funding height of the last selected candidate, then the final loads satisfy
\begin{equation}\label{eq:load-invariants}
  \sum_{i\in\mathcal{N}}w_i\ell_i=k,
  \qquad 0\le\ell_i\le L,
  \qquad L\ge\frac{k}{n}.
\end{equation}
\end{lemma}

\begin{proof}
For each candidate $c$, the function $F_{c,\ell}(x)$ is continuous, is zero for sufficiently small $x$, and tends to infinity as $x$ increases.
Once positive, it is strictly increasing.
Equation~\eqref{eq:funding-height} therefore has a unique solution.
Suppose an update changes the load vector from $\ell$ to $\ell'$.
Since $\ell'_i\ge\ell_i$ for every voter, every remaining candidate $d$ satisfies $F_{d,\ell'}(x)\le F_{d,\ell}(x)$ for all $x$.
Its funding height therefore cannot decrease.
The selected candidate had minimum height before the update, so the sequence of selected heights is nondecreasing.
By~\eqref{eq:funding-height}, each update adds one unit of total weighted load, so the total after $k$ steps is $k$.
Every individual load is zero or a previously selected height, and hence is at most $L$.
Thus $k=\sum_iw_i\ell_i\le nL$, proving \eqref{eq:load-invariants}.
\end{proof}


For rationally encoded voter weights, Algorithm~\ref{alg:seq-phragmen} can be implemented in polynomial time: each funding height is obtained by sorting the current loads of the candidate's approvers and solving the resulting piecewise-linear equation.
Our main result is the following.

\begin{theorem}\label{thm:main}
Under arbitrary tie-breaking, weighted sequential Phragm\'en returns a $2$-stable size-$k$ committee for every finite weighted approval election.
\end{theorem}

\section{The 2-Stability Guarantee}\label{sec:proof}

For each voter, we consider the gaps between the funding heights of her approved winners, including the initial and final gaps.
For any candidate not in $W$, its approvers could not have funded it before the next selected candidate.
Integrating this constraint gives the bound in Lemma~\ref{lem:credit-area}.
We then consider an alternative committee $T$ and separate the candidates in $T\cap W$ from those in $T\setminus W$.
The payments for the former and the bound for the latter together limit the total weight of voters who prefer $T$.
Combining this limit with $L\ge\frac{k}{n}$ proves the theorem.


Fix a run of the algorithm, and let $c_1,\ldots,c_k$ be the candidates selected in this order.
Write
\[
  0<\theta_1\le\theta_2\le\cdots\le\theta_k=L
\]
for their funding heights.
For each voter $i$, let
\[
  p_{i,1}<p_{i,2}<\cdots<p_{i,u_i(W)}
\]
be the rounds in which an approved candidate is selected, and define
\[
  \tau_{i,0}=0,
  \qquad
  \tau_{i,j}=\theta_{p_{i,j}}\quad(1\le j\le u_i(W)),
  \qquad
  \tau_{i,u_i(W)+1}=L.
\]
Thus $\tau_{i,j}$ is voter $i$'s load after her $j$th approved winner is selected, with $0$ and $L$ as the initial and final endpoints.
Define the gaps
\begin{equation*}
  d_{i,j}=\tau_{i,j}-\tau_{i,j-1}
  \quad(j=1,\ldots,u_i(W)+1).
\end{equation*}
The first $u_i(W)$ gaps are the load increases when approved candidates are selected, and the last is $L$ minus the voter's final load.
They are nonnegative and sum to $L$.
Some gaps may be zero, and if $u_i(W)=0$, the sole gap has length $L$.
Immediately before $c_{p_{i,j}}$ is selected, voter $i$ has load $\tau_{i,j-1}$.
Her payment toward that candidate is therefore $w_i d_{i,j}$.
Consequently, for every selected candidate, the weighted gaps ending at that candidate sum to exactly one over all its approvers.

For $z$ outside the finite set of selected heights, define the difference between $z$ and voter $i$'s most recent load by
\begin{equation*}
  \alpha_i(z)=z-\max\bigl(
    \{0\}\cup
    \{\tau_{i,j}:1\le j\le u_i(W),\ \tau_{i,j}<z\}
  \bigr).
\end{equation*}
Its values at the selected heights may be chosen arbitrarily.
On every gap it rises linearly from zero to the gap length.
Hence
\begin{equation}\label{eq:ramp-area}
  \int_0^L\alpha_i(z)\,dz
  =\frac{1}{2}\sum_{j=1}^{u_i(W)+1}d_{i,j}^2.
\end{equation}
We use this identity to bound the squared gaps for voters approving candidates outside $W$, counting each voter once for each such candidate she approves.

\begin{lemma}\label{lem:credit-area}
Let $W$, $L$, and $d_{i,j}$ be as above.
Let $R\subseteq\mathcal{C}\setminus W$ and $S\subseteq\mathcal{N}$ be nonempty sets.
If $M=\sum_{i\in S}w_i u_i(R)>0$, then
\begin{equation}\label{eq:credit-area}
  \sum_{i\in S}w_i u_i(R)
    \left(\sum_{j=1}^{u_i(W)+1}d_{i,j}^2\right)
  <2|R|L.
\end{equation}
\end{lemma}

\begin{proof}
Fix $z\in(0,L)$ distinct from every selected height, and let $p$ be the number of selected heights below $z$.
After round $p$, the next selected candidate has funding height $\theta_{p+1}>z$.
Every $c\in\mathcal{C}\setminus W$ is still available, so its funding height is at least $\theta_{p+1}$.
Evaluating its funding function at $z$ therefore gives
\begin{equation}\label{eq:candidate-credit}
  \sum_{i:c\in A_i}w_i\alpha_i(z)\le1.
\end{equation}
Here $\alpha_i(z)$ is the difference between $z$ and voter $i$'s current load.
The inequality therefore holds for every unelected candidate and almost every $z\in[0,L]$.
Sum~\eqref{eq:candidate-credit} over $c\in R$, discard voters outside $S$, and change the order of summation.
Since $u_i(R)=|A_i\cap R|$,
\begin{equation*}
  \sum_{i\in S}w_i u_i(R)\alpha_i(z)\le \sum_{c\in R}\sum_{i:c\in A_i}w_i\alpha_i(z)\le |R|
  \quad\text{for almost every }z\in[0,L].
\end{equation*}
Before the first selected height, every $\alpha_i(z)$ equals $z$, and hence the left-hand side equals $Mz$.
Set \( \varepsilon=\min\big\{\frac{\theta_1}{2},\frac{|R|}{2M}\big\}>0. \) Then $Mz\le \frac{|R|}{2}$ on $(0,\varepsilon)$.
Therefore
\[
  \int_0^L\left(|R|-\sum_{i\in S}w_i u_i(R)\alpha_i(z)\right)\,dz
  \ge \int_0^\varepsilon\bigl(|R|-Mz\bigr)\,dz
  \ge \frac{|R|\varepsilon}{2}>0.
\]
It follows that
\[
  \int_0^L\sum_{i\in S}w_i u_i(R)\alpha_i(z)\,dz<|R|L.
\]
Substituting \eqref{eq:ramp-area} and multiplying by two proves \eqref{eq:credit-area}.
\end{proof}

Suppose a group $S$ strictly prefers an alternative committee that retains the candidates in $H\subseteq W$ and adds candidates in $R\subseteq\mathcal C\setminus W$.
Each voter in $S$ must approve at least one more candidate in $R$ than in $W\setminus H$.
The next lemma uses this observation to bound $L\wt(S)$.

\begin{lemma}\label{lem:blocker-load}
Let $W$, $L$, and $d_{i,j}$ be as above.
Let $H\subseteq W$, $\varnothing\ne R\subseteq\mathcal{C}\setminus W$, and $\varnothing\ne S\subseteq\mathcal{N}$.
If every $i\in S$ satisfies
\begin{equation}\label{eq:blocker-antecedent}
  u_i(R)\ge u_i(W\setminus H)+1,
\end{equation}
then
\begin{equation*}
  L\wt(S)<|H|+\sqrt{2|R|L\wt(S)}.
\end{equation*}
\end{lemma}

\begin{proof}
For every $i\in\mathcal N$, define
\[
  \mathcal E_i=
  \{j\in\{1,\ldots,u_i(W)\}:c_{p_{i,j}}\in H\}
\]
to index the gaps ending at selected candidates in $H$ that voter $i$ approves.
For every $i\in S$, also define
\[
  \mathcal G_i=
  \{j\in\{1,\ldots,u_i(W)\}:c_{p_{i,j}}\in W\setminus H\}
  \cup\{u_i(W)+1\}.
\]
$\mathcal G_i$ indexes the gaps ending at approved winners in $W\setminus H$, together with the final gap to $L$.
We bound the gaps in $\mathcal E_i$ by the payments for $H$ and those in $\mathcal G_i$ by Lemma~\ref{lem:credit-area}.
These two sets partition the gap indices, so
\begin{equation}\label{eq:gap-partition}
  \sum_{j\in\mathcal E_i}d_{i,j}
  +\sum_{j\in\mathcal G_i}d_{i,j}=L.
\end{equation}
Since $|\mathcal G_i|=u_i(W\setminus H)+1$, Cauchy--Schwarz and \eqref{eq:blocker-antecedent} give
\begin{equation}\label{eq:individual-omitted}
  \left(\sum_{j\in\mathcal G_i}d_{i,j}\right)^2
  \le \bigl(u_i(W\setminus H)+1\bigr)
      \sum_{j\in\mathcal G_i}d_{i,j}^2
  \le u_i(R)\sum_{j=1}^{u_i(W)+1}d_{i,j}^2.
\end{equation}

For each candidate in $H$, the weighted gaps ending at that candidate are its actual load payments and sum to one over all approving voters.
Restricting to $S$ and summing over $H$ yields
\begin{equation}\label{eq:retained-bound}
  \sum_{i\in S}w_i\sum_{j\in\mathcal E_i}d_{i,j}
  \le \sum_{i\in\mathcal N}w_i\sum_{j\in\mathcal E_i}d_{i,j}
  =|H|.
\end{equation}
Weighted Cauchy--Schwarz, \eqref{eq:individual-omitted}, and Lemma~\ref{lem:credit-area} imply
\begin{align}
  \left(\sum_{i\in S}w_i
    \sum_{j\in\mathcal G_i}d_{i,j}\right)^2
  &\le \wt(S)\sum_{i\in S}w_i
    \left(\sum_{j\in\mathcal G_i}d_{i,j}\right)^2 \notag\\
  &\le \wt(S)\sum_{i\in S}w_i u_i(R)
       \left(\sum_{j=1}^{u_i(W)+1}d_{i,j}^2\right) \notag\\
  &<2|R|L\wt(S).\label{eq:omitted-bound}
\end{align}
Here Lemma~\ref{lem:credit-area} applies since $u_i(R)\ge1$ for every $i\in S$.
Finally, summing~\eqref{eq:gap-partition} with weights over $S$ and using \eqref{eq:retained-bound} and \eqref{eq:omitted-bound} yields
\begin{align*}
  L\wt(S)
  &=\sum_{i\in S}w_i\sum_{j\in\mathcal E_i}d_{i,j}
    +\sum_{i\in S}w_i\sum_{j\in\mathcal G_i}d_{i,j}<|H|+\sqrt{2|R|L\wt(S)}.
\end{align*}
\end{proof}

We will use the following inequality to obtain a contradiction from Lemma~\ref{lem:blocker-load}.

\begin{lemma}\label{lem:scalar}
For $x,y,z\ge0$ with $x+y>0$,
\begin{equation*}
  z\ge2(x+y)
  \quad\Longrightarrow\quad
  x+\sqrt{2yz}\le z.
\end{equation*}
Equality is possible only when $x=0$ and $z=2y$.
\end{lemma}

\begin{proof}
Since $z\ge2(x+y)>0$, we have $z-x>0$ and
\[
  (z-x)^2-2yz
  =z\bigl(z-2(x+y)\bigr)+x^2\ge0.
\]
For equality, both nonnegative terms must vanish.
Since $z>0$, this forces $x=0$ and $z=2(x+y)=2y$.
\end{proof}

We now combine the lemmas to prove Theorem~\ref{thm:main}.

\begin{proof}[Proof of Theorem~\ref{thm:main}]
Let $L$ be the final height.
Fix a nonempty alternative $T$ and put
\[
  S=\{i\in\mathcal{N}:u_i(T)>u_i(W)\}.
\]
If $S=\varnothing$, the conclusion is immediate.
Suppose for a contradiction that $T$ is a $2$-blocking deviation.
Then
\begin{equation}\label{eq:block-assumption}
  \wt(S)\ge\frac{2n|T|}{k}.
\end{equation}
Partition $T$ into \( H=T\cap W\) and \(R=T\setminus W \) so that $|T|=|H|+|R|$.
Because $S$ is nonempty, $R$ is nonempty: if $R$ were empty, then $T\subseteq W$ and no voter could strictly prefer $T$.

For every $i\in S$, cancellation of the common part $H$ gives
\[
  u_i(T)-u_i(W)
  =u_i(R)-u_i(W\setminus H)>0.
\]
Approval utilities are integral, and therefore
\begin{equation}\label{eq:integral-cancellation}
  u_i(R)\ge u_i(W\setminus H)+1.
\end{equation}

Lemma~\ref{lem:blocker-load}, applied with \eqref{eq:integral-cancellation}, gives
\begin{equation}\label{eq:strict-side}
  L\wt(S)<|H|+\sqrt{2|R|L\wt(S)}.
\end{equation}
On the other hand, Lemma~\ref{lem:invariants} and \eqref{eq:block-assumption} imply
\begin{equation*}
  L\wt(S)
  \ge\frac{k}{n}\frac{2n|T|}{k}
  =2|T|=2(|H|+|R|).
\end{equation*}
Lemma~\ref{lem:scalar} now yields $|H|+\sqrt{2|R|L\wt(S)}\le L\wt(S)$, contradicting~\eqref{eq:strict-side}.
Thus no nonempty $T$ can be a $2$-blocking deviation.
\end{proof}

Together with the polynomial-time implementation of the rule, Theorem~\ref{thm:main} gives the following corollary.

\begin{corollary}\label{cor:nonempty}
The $2$-core of every finite weighted approval election is nonempty.
Moreover, for rationally encoded voter weights, a $2$-stable committee can be computed in polynomial time using weighted sequential Phragm\'en.
\end{corollary}

\section{Tightness}\label{sec:tightness}

We show that the factor $2$ cannot be improved for sequential Phragm\'en, even when all voters have equal weight.
The construction adapts the cyclic selection of candidates in Skowron's analysis of proportionality degree~\cite[Proposition~1]{skowronproportionality}.
Here, every voter in the deviating group receives the same number of approved winners, so all members strictly prefer the alternative committee.

\begin{theorem}\label{thm:tightness}
For every $1\le\lambda<2$, there is an unweighted approval election in which sequential Phragm\'en returns a committee that is not $\lambda$-stable, regardless of tie-breaking.
\end{theorem}

\begin{proof}
Fix integers $q\ge2$ and $r\ge1$, and let $k=qr$.
There are $q$ voters $s_1,\ldots,s_q$, a further $B=\frac{q(q-1)}{2}+1$ voters, and $q-1$ voters $h_2,\ldots,h_q$.
All voters have unit weight, so the total number of voters is
\[
 n=B+2q-1=\frac{q(q+3)}{2}.
\]
There are $qr$ candidates $b_{a,j}$, indexed by $1\le a\le r$ and $1\le j\le q$, and $r+1$ candidates $d_1,\ldots,d_{r+1}$.
The approval sets are as follows.
\begin{center}
\begin{tabular}{ll}
\hline
Voters & Approved candidates \\
\hline
$s_j$, for $1\le j\le q$ & $b_{1,j},\ldots,b_{r,j}$ and all $d$ candidates \\
Each of the $B$ further voters & All $b$ candidates \\
$h_j$, for $2\le j\le q$ & $b_{1,1},\ldots,b_{1,j-1}$ \\
\hline
\end{tabular}
\end{center}
We prove that the rule selects all $qr$ candidates $b_{a,j}$ and none of the $d$ candidates.
It selects one candidate from each group $j=1,\ldots,q$ in this order, repeating the cycle $r$ times.
The selected funding heights are $\delta,2\delta,\ldots,qr\delta$, where
\[
 \delta=\frac{1}{B+q}=\frac{2}{q(q+1)+2}.
\]
To verify this order, we evaluate the funding function $F_{c,\ell}$ at each proposed height.
For unit weights, a supporter whose current load is $\ell_i$ contributes $\max\{0,x-\ell_i\}$ at height $x$.

Consider the first cycle, and suppose the first $j-1$ candidates have been selected at the claimed heights.
At height $j\delta$, voter $s_j$ still has load zero and therefore contributes $j\delta$ to $b_{1,j}$.
Each of the $B$ common supporters and each of the $q-j$ voters $h_{j+1},\ldots,h_q$ contributes $\delta$.
For $j=1$, their loads are zero; for $j>1$, they all supported the preceding candidate and have load $(j-1)\delta$.
Consequently,
\[
 F_{b_{1,j},\ell}(j\delta)=j\delta+(B+q-j)\delta=1.
\]
For a remaining candidate $b_{1,h}$ with $h>j$, the same calculation gives
\[
 F_{b_{1,h},\ell}(j\delta)=j\delta+(B+q-h)\delta<1.
\]
Any candidate $b_{a,h}$ with $a\ge2$ has only its voter $s_h$ and the $B$ common supporters, so its funding function at this height is at most $(B+j)\delta$.
This is less than $1$ when $j<q$.
When $j=q$, equality holds only for group $q$, whose candidates all have identical supporters.
Thus we may name the selected candidate $b_{1,q}$ without affecting the election.

It remains to check that no $d$ candidate is selected during this cycle.
Immediately before the $j$th selection, voters $s_1,\ldots,s_{j-1}$ have loads $\delta,\ldots,(j-1)\delta$, and the other voters in this group have load zero.
Every $d$ candidate therefore has
\[
 F_{d_a,\ell}(j\delta)
 =\delta\left(qj-\frac{j(j-1)}{2}\right)
 \le\frac{q(q+1)\delta}{2}
 =\frac{q(q+1)}{q(q+1)+2}<1.
\]
All other candidates have funding heights strictly above $j\delta$, except for the identical candidates in group $q$ at the last step.
This proves the claimed order and heights for the first cycle.

After the first cycle, every candidate approved by any $h_j$ has been selected.
Each remaining $b$ candidate has exactly its voter $s_j$ and the $B$ common supporters.
Suppose the cyclic order has continued up to some step, and consider the next proposed height $t\delta$.
The $q$ voters $s_j$ last had their loads updated at the preceding $q$ selection heights.
Their contributions at $t\delta$ are therefore $\delta,2\delta,\ldots,q\delta$.
The voter whose turn comes next contributes $q\delta$, while each common supporter contributes $\delta$.
The candidates in this voter's group have total contribution
\[
 q\delta+B\delta=1.
\]
Every other remaining $b$ candidate has total contribution at most $(B+q-1)\delta<1$.
Every $d$ candidate again has total contribution $\frac{q(q+1)\delta}{2}<1$.
Thus the next candidate belongs to the claimed group and is selected at height $t\delta$.
Induction proves the order for all $qr$ selections.
The only ties are between candidates with identical supporters in the same group.

Let $W$ be the resulting committee, $S=\{s_1,\ldots,s_q\}$, and $T=\{d_1,\ldots,d_{r+1}\}$.
Every voter in $S$ has $u_i(W)=r$ and $u_i(T)=r+1$, so all of them strictly prefer $T$.
The ratio of the size of this group to the proportional share of $T$ is
\[
 \frac{|S|k}{n|T|}=\frac{2qr}{(q+3)(r+1)}.
\]
Taking $q=r=m$ and letting $m$ grow gives
\[
 \frac{|S|k}{n|T|}=\frac{2m^2}{(m+3)(m+1)}\longrightarrow2.
\]
For any $1\le\lambda<2$, choosing $m$ sufficiently large yields $|S|\ge\lambda\frac{n|T|}{k}$.
Hence $T$ is a $\lambda$-blocking deviation, as required.
\end{proof}

\section{Discussion and open questions}\label{sec:discussion}

We proved that weighted sequential Phragm\'en always returns a $2$-stable committee, for all positive voter weights and arbitrary tie-breaking.
The proof follows the loads assigned by the rule and bounds the total weight of voters who prefer an alternative committee.
This improves the $3.651$-stability guarantee of Gao, Sun, and Vondr\'ak~\cite{gao2025} to $2$.
In concurrent work, Becker, Greger, and Peters~\cite{becker2026core} settled the exact-core existence question.
Our guarantee and tightness result show that $2$ is the best uniform stability guarantee of sequential Phragm\'en itself.

A natural direction is to seek a simple voting rule that improves upon the $6.24$-stability guarantee of Song and Nguyen~\cite{song2026ordinal} for general monotone preferences.
Sequential Phragm\'en uses approval sets to determine which voters share each candidate's load, so the rule itself would need to be generalized before our approach could be applied to this setting.

\paragraph{AI Disclosure.}
Codex was used to assist with proof auditing, literature checking, and drafting the manuscript.
The authors are responsible for independently verifying the correctness, originality, and accuracy of all content.

\end{document}